\documentclass[12pt]{article}
\usepackage[latin1]{inputenc}
\usepackage{cmap}
\usepackage{lmodern}

\usepackage{amssymb, amsmath, amsthm}
\usepackage[a4paper,top=25mm,bottom=25mm,left=25mm,right=25mm]{geometry}
\usepackage{etex}
\usepackage{ragged2e}

\usepackage{authblk} % for headings
\usepackage{pifont}
\usepackage{graphicx}
\usepackage[usenames,dvipsnames,svgnames,table]{xcolor}
\usepackage[figuresright]{rotating}
\usepackage{xtab} % tackle the long tables
\usepackage{longtable} % tackle the long tables
\usepackage{multirow}
\usepackage{footnote}
\usepackage[stable]{footmisc}
\usepackage{chngpage} % allows for temporary adjustment of side margins
\usepackage{pdflscape} % landscape environment
\usepackage{tocbibind} % includes everything in ToC

\usepackage{pgfplots}
\pgfplotsset{compat=1.14}
\usepackage{setspace}

\makesavenoteenv{tabular}
\usepackage{tabularx}
\usepackage{booktabs}
\usepackage{threeparttable}
\usepackage[referable]{threeparttablex} % footnotes in tabu
\newcolumntype{R}{>{\raggedleft\arraybackslash}X}
\newcolumntype{L}{>{\raggedright\arraybackslash}X}
\newcolumntype{C}{>{\centering\arraybackslash}X}
\newcolumntype{A}{>{\columncolor{gray!25}}C}
\newcolumntype{a}{>{\columncolor{gray!25}}c}

\newlength{\tablen}

\usepackage{dcolumn} % alignment to decimal points
\newcolumntype{.}{D{.}{.}{-1}}

\usepackage{tikz}
\usetikzlibrary{arrows}
\usepackage[semicolon]{natbib}
\usepackage[hyphens]{url}
\usepackage{hyperref} % [hidelinks]
\hypersetup{
  colorlinks   = true,    % Colours links instead of ugly boxes
  urlcolor     = blue,    % Colour for external hyperlinks
  linkcolor    = blue,    % Colour of internal links
  citecolor    = red      % Colour of citations
}
\usepackage{microtype}
\usepackage[justification=centering]{caption} % [justification=centering]

\usepackage[labelformat=simple]{subcaption}

\DeclareCaptionLabelFormat{parenthesis}{(#2)}
\makeatletter
\renewcommand\p@subfigure{\arabic{figure}.}
\makeatother

\DeclareCaptionLabelFormat{parenthesis}{(#2)}
\makeatletter
\renewcommand\p@subtable{\arabic{table}.}
\makeatother

\usepackage{enumitem}

\setlist[itemize]{leftmargin=2.5\parindent}
\setlist[enumerate]{leftmargin=2.5\parindent}

\theoremstyle{plain}

\newtheorem{proposition}{Proposition}%[section]

\theoremstyle{definition}

\newtheorem{definition}{Definition}%[section]

\theoremstyle{remark}

\def\keywords{\vspace{.5em} % Add keywords
{\noindent \textit{Keywords}:\,}}

\def\JEL{\vspace{.5em} % Add keywords
{\noindent \textbf{\emph{JEL} classification number}:\,}}

\def\AMS{\vspace{.5em} % Add keywords
{\noindent \textbf{\emph{AMS} classification number}:\,}}

\author{\href{https://sites.google.com/site/laszlocsato87}{L\'aszl\'o Csat\'o}\thanks{~e-mail: laszlo.csato@uni-corvinus.hu} }
\affil{Institute for Computer Science and Control, Hungarian Academy of Sciences (MTA SZTAKI) \\
Laboratory on Engineering and Management Intelligence, Research Group of Operations Research and Decision Systems}
\affil{Corvinus University of Budapest (BCE) \\
Department of Operations Research and Actuarial Sciences}
\affil{Budapest, Hungary}
\title{Was Zidane honest or well-informed? \\
How UEFA barely avoided a serious scandal}
\date{\today}

\def\Dedication{ % Add keywords
{\noindent $\mathfrak{M\ddot{o}gliche}$ $\mathfrak{Gefechte}$ $\mathfrak{sind}$ $\mathfrak{der}$ $\mathfrak{Folgen}$ $\mathfrak{wegen}$ $\mathfrak{als}$ $\mathfrak{wirkliche}$ $\mathfrak{zu}$ $\mathfrak{betrachten}$}
\vspace{0.25cm}
\flushright
\noindent (Carl von Clausewitz: \emph{Vom Kriege})
\vspace{1cm} 
\justify}

\begin{document}

\maketitle

\Dedication

\begin{abstract}
UEFA European Championship 1996 qualification is known to violate strategy-proofness. It has been proved recently that a team could be better off by exerting a lower effort: it might be optimal to concede some goals in order to achieve a better position among runners-up, and hence avoid a hazardous play-off.
We show that it is not only an irrelevant scenario with a marginal probability since France had an incentive to kick two own goals on its last match against Israel.

\JEL{C44, D71, L83} % Operations Research, Statistical Decision Theory  % Social Choice, Clubs, Committees, Associations % Sports

\AMS{91B14} % Social choice

\keywords{OR in sport; UEFA Euro 1996; ranking; strategy-proofness}
% soccer; 2018 FIFA World Cup; UEFA; axiomatic approach; manipulation
\end{abstract}

%yannick.berker@dkfz.de

\section{Introduction} \label{Sec1}

Suppose you are \emph{\href{https://en.wikipedia.org/wiki/Zinedine_Zidane}{Zin\'edine Zidane}},\footnote{~Zin\'edine Yazid Zidane is a French retired attacking midfielder and the current manager of Real Madrid. He was named the best European footballer of the past 50 years in the UEFA Golden Jubilee Poll in 2004 and is regarded as one of the greatest players of all time.} a player of the French national association football team. You are playing against Israel in Group 1 of the 1996 UEFA European Championship qualification tournament, and your teammate Lizarazu has scored a goal recently, in the 89th minute of the match. What should you do? It will be revealed that the optimal course of action is to kick two goals -- in your own net!

According to our knowledge, there were two football matches where a team deliberately kicked an own goal because of some strange incentives.
Perhaps the most famous example is \href{https://en.wikipedia.org/wiki/Barbados_4\%E2\%80\%932_Grenada_(1994_Caribbean_Cup_qualification)}{Barbados vs Grenada (1994 Caribbean Cup qualification)}, played on January 27, 1994 \citep[Section~3.9.4]{KendallLenten2017}. In the qualifiers, each match must have a winner to be decided in a sudden-death 30 minutes extra time with a golden goal counting as two.
Barbados needed to win by two goals to progress, otherwise Grenada would advance.
Barbados had a 2-1 lead in the 87th minute, when the players realized that by scoring an own goal, they would have much more time remaining to win by two goals. After choosing this radical strategy, they succeeded in preserving the 2-2 result. Finally, Barbados have managed to score a goal in extra time and qualified.
FIFA did not penalize the Barbados Federation since they behaved according to the prevailing rules. Nevertheless, this controversial rule was never used again.

The second case was \href{https://en.wikipedia.org/wiki/1998_AFF_Championship\#Controversy}{Thailand vs Indonesia in the 1998 Tiger Cup} on 31 August 1998 \citep[Section~3.9.2]{KendallLenten2017}.
Both teams were already qualified for the semi-finals such that the group-winner would face hosting Vietnam, while the runner-up would play against Singapore, perceived to be the easier opponent, and would avoid moving the team's training base.
The score was 2-2 after 90 minutes implying Indonesia would be the group-winner. However, in extra time an Indonesian defender deliberately scored an own goal despite the Thai's attempts to stop him doing so.
FIFA fined both teams $40,000$ for '\emph{violating the spirit of the game}', furthermore, the Indonesian defender was banned from domestic football for one year and international football for life.

On the other hand, it has been shown recently that tournaments with subsequent group stages, where matches played in the preliminary round between players who qualify to the next round are carried over, suffer from such perverse incentives, and there existed at least two matches where a team had an incentive to win by few goals difference, possibly by throwing own goals \citep{Csato2017k}. However, we do not know about the use of this format in soccer.

Besides a careful analysis of historical events, we think it is also important to consider \emph{possible} scandals since rules may have a low probability to go awry, but the potential costs can be enormous. \citet{DagaevSonin2013}, \citet{DagaevSonin2017}, \citet{Csato2017h} and \citet{Csato2018b} have shown such hypothetical examples. \citet{Csato2017i} proved by a theoretical model that 1996 UEFA European Championship qualification was incentive incompatible, in other words, a team might be strictly better off by playing a draw (or even losing) instead of winning.
In the following, we will describe just such a situation at the first time, which occurred in the match France vs Israel, played on 15 November 1995.
Since the manipulation by France might have hurt a third team, it would be a more serious case than Barbados vs Grenada.
%The paper is structured as follows.

The rest of the paper is organized as follows. Section~\ref{Sec2} presents the qualification rules and outlines their theoretical model. The particular situation is detailed in Section~\ref{Sec3}, and we conclude in Section~\ref{Sec4}.

\section{Theoretical background} \label{Sec2}

\href{https://en.wikipedia.org/wiki/UEFA_Euro_1996_qualifying}{1996 UEFA European Championship qualification} divided $47$ teams into seven groups of six and one group of five teams. All teams played a home-and-away round-robin tournament in their groups. The group winners along with the six best runners-up qualified automatically, while the two worst runners-up were involved in a play-off.

This was the first time when the qualification for the European Football Championship awarded three points (instead of two) for a win. Tie-breaking rules in the groups were as follows:
\begin{enumerate}%[label=\alph*)]
\item
greater number of points obtained in all group matches (three points for a win, one for a draw and no points for a defeat);
\item
greater number of points obtained in the matches played between the teams concerned;
\item
superior goal difference in the matches played between the teams concerned;
\item
greater number of goals scored away from home in the matches played between the teams concerned.
\end{enumerate}
Further tie-breaking rules in the groups are not relevant for us. 

The runners-up were placed in a separate table to rank them such that only matches played against the first, third and fourth teams of the groups were regarded. Tie-breaking rules among the second-placed teams were as follows:
\begin{enumerate}%[label=\alph*)]
\item
greater number of points obtained in the group matches considered;
\item
superior goal difference in the group matches considered;
\item
greater number of goals scored in the group matches considered.
\end{enumerate}
Further tie-breaking rules are not relevant for us.

The difference of ranking in the groups and ranking among second-placed teams (from different groups) is responsible for the incentive incompatibility of the qualification. In the latter case, the matches played against the fifth, and -- in some groups -- against the sixth teams are discarded, therefore a runner-up can influence its performance measures among second-placed teams by changing the fourth and fifth teams in its group.
Due to the monotonicity of group ranking, a team cannot achieve a better position in its group by kicking an own goal but it may gain some goals or even points in the comparison of runners-up, provided that it remains the second.
Hence, it may qualify automatically for the final tournament as one of the six best second-placed teams, instead of playing a risky play-off match.

This argument can be applied in the theoretical analysis of a similar \emph{group-based qualification system} $\mathcal{Q}$ where participating teams are divided into $k$ groups such that the number of teams is $n_i$ in group $i =1, 2, \dots ,k$.
Furthermore, under a given set of match results $\mathbf{R}$, the top $a_i \geq 0$ teams in group $i$ directly qualify, while the next $b_i \geq 0$ teams -- whose set is denoted by $B_i(\mathbf{R})$ -- are compared in a so-called \emph{extra group}. The remaining $n_i - a_i - b_i \geq 0$ teams are eliminated.

Let $B(\mathbf{R}) = \sum_{i=1}^k B_i(\mathbf{R})$ be the set of teams of the extra group. In the extra group, only the group matches against the top $c_i$ teams are considered, where $a_i + b_i \leq c_i \leq n_i$. Suppose that direct qualification is preferred to advancing to play-offs, which is better than being eliminated. These three sets of teams are denoted by $T_1(\mathbf{R})$, $T_2(\mathbf{R})$, and $T_3(\mathbf{R})$, respectively.

Let $\mathcal{Q}$ be a group-based qualification system with the set of match results $\mathbf{R}$ and $x$ be a team. The set of match results $\mathbf{R_{x}'}$ is said to be \emph{disadvantageous} for team $x$, if $\mathbf{R}$ and $\mathbf{R_{x}'}$ are identical except for some opponents of team $x$ have scored more goals against it in $\mathbf{R_{x}'}$.

\begin{definition} \label{Def1}
\emph{Manipulation}:
Team $x$ can \emph{manipulate} the group-based qualification system $\mathcal{Q}$ under the set of match results $\mathbf{R}$ if there exists a set of match results $\mathbf{R_{x}'}$, which is disadvantageous for team $x$ such that one of the following holds:
(1) $x \in T_3(\mathbf{R})$ but $x \in T_1(\mathbf{R_{x}'}) \cup T_2(\mathbf{R_{x}'})$; or
(2) $x \in T_2(\mathbf{R})$ but $x \in T_1(\mathbf{R_{x}'})$.
\end{definition}

In other words, team $x$ is strictly better off after its manipulation despite the set of match results is disadvantageous for it.

This setting may allow for incentive incompatibility.

\begin{proposition} \label{Prop1}
Let $\mathcal{Q}$ be a group-based qualification system. There exists a set of match results $\mathbf{R}$ and a team $x \in B(\mathbf{R})$ that can manipulate the qualification under $\mathbf{R}$ if:
\begin{itemize}
\item
the number of groups is at least $k \geq 2$;
\item
there is a difference in the allocation of teams in the extra group, that is, at least two of the sets $B(\mathbf{R}) \cap T_1(\mathbf{R})$, $B(\mathbf{R}) \cap T_2(\mathbf{R})$, and $B(\mathbf{R}) \cap T_3(\mathbf{R})$ are non-empty;
\item
there are more teams than $n_i - c_i \geq 1$ behind team $x$ in its group $i$.
\end{itemize}
\end{proposition}

\begin{proof}
It is covered by the model described in \citet{Csato2017i}, specifically, see \citet[Theorem~1]{Csato2017i}.
The proof is based on the idea presented above, that is, a team can improve its position in the extra group -- by influencing its set of matches taken in the extra group into account -- through conceding some goals, which is impossible if all matches against teams behind team $x$ are discarded or counted in the extra group. It makes the last condition of Proposition~\ref{Prop1} necessary.
The second requirement guarantees that achieving a better position among the teams of the extra group is effective with respect to qualification.
\end{proof}

Further discussion of Proposition~\ref{Prop1} and conditions providing strategy-proofness can be found in \citet{Csato2017i}.

\section{The potential scandal} \label{Sec3}

1996 UEFA European Championship qualification is covered by the model presented in Section~\ref{Sec2}. The top team directly qualified to the tournament ($a_i = 1$), while the next team were placed into an extra group for each group $i = 1,2, \dots ,8$ ($b_i = 1$).
The top six runners-up also qualified, however, the two worst were involved in a play-off. Consequently, $\left| B(\mathbf{R}) \cap T_1(\mathbf{R}) \right| = 6$ and $\left| B(\mathbf{R}) \cap T_2(\mathbf{R}) \right| = 2$.
Finally, $c_i = 4$ and $n_i = 6$ with the exception of $n_3 = 5$, so there are two teams more than $n_i - c_i$ behind any second-placed team (the third and the fourth) in its group $i$.
Hence the conditions of Proposition~\ref{Prop1} are satisfied, implying the incentive incompatibility of the qualification: there exists a set of match results under which an arbitrary runner-up can manipulate it.

Now we will show that the qualification tournament was very close to such a situation.

\begin{table}[ht]
\begin{threeparttable}
\centering
\caption{1996 UEFA European Championship qualification -- Group 1: Final standing before the last match France against Israel}
\label{Table1}
\rowcolors{1}{}{gray!20}
    \begin{tabularx}{\linewidth}{Cl CCCC CC >{\bfseries}C} \toprule \showrowcolors
    Pos   & Team  & W     & D     & L     & GF    & GA    & GD    & Pts \\ \hline
    1     & Romania & 6     & 3     & 1     & 18    & 9     & 9     & 21 \\
    2     & France & 4     & 5     & 0     & 20    & 2     & 18    & 17 \\
    3     & Slovakia & 4     & 2     & 4     & 14    & 18    & -4    & 14 \\
    4     & Poland & 3     & 4     & 3     & 14    & 12    & 2     & 13 \\
    5     & Israel & 3     & 3     & 3     & 13    & 11    & 2     & 12 \\
    6     & Azerbaijan & 0     & 1     & 9     & 2     & 29    & -27   & 1 \\ \bottomrule  
    \end{tabularx}
    
    \begin{tablenotes}
    \item
\footnotesize{Pos = Position; W = Won; D = Drawn; L = Loss; GF = Goals for; GA = Goals against; GD = Goal difference; Pts = Points. Romania, Slovakia, Poland and Azerbaijan have played 10, France and Israel 9 matches.}
    \end{tablenotes}
\end{threeparttable}
\end{table}

Table~\ref{Table1} shows the standing of \href{https://en.wikipedia.org/wiki/UEFA_Euro_1996_qualifying_Group_1}{Group 1} before the last match of France against Israel, played on 15 November 1995 at 20:45 according to Central European Time in Caen, France. The other two matches of Group 1 in the last matchday were started at 16:00 and 17:30 on the same day, so Table~\ref{Table1} can be assumed as common knowledge at the beginning of this particular match. France led by a score of 2-0 at the 89th minute. What were the options of the team?

\begin{table}[ht]
\begin{threeparttable}
\centering
\caption{1996 UEFA European Championship qualification -- Group 1: Match results}
\label{Table2}
\rowcolors{1}{}{gray!20}
    \begin{tabularx}{\linewidth}{cl CCC CCC} \toprule
    Position   & Team      & 1     & 2     & 3     & 4     & 5     & 6 \\ \hline
    1     & Romania & ---     & 1-3   & 3-2   & 2-1   & 2-1   & 3-0 \\
    2     & France & 0-0   & ---     & 4-0   & 1-1   & \textbf{?}   & 10-0 \\
    3     & Slovakia & 0-2   & 0-0   & ---     & 4-1   & 1-0   & 4-1 \\
    4     & Poland & 0-0   & 0-0   & 5-0   & ---     & 4-3   & 1-0 \\
    5     & Israel & 1-1   & 0-0   & 2-2   & 2-1   & ---     & 2-0 \\
    6     & Azerbaijan & 1-4   & 0-2   & 0-1   & 0-0   & 0-2   & --- \\ \bottomrule
    \end{tabularx}
    
    \begin{tablenotes}
    \item
\footnotesize{Position is given without the match France against Israel \\
Home team is in the row, away team (represented by its position) is in the column}
    \end{tablenotes}
\end{threeparttable}
\end{table}

France is guaranteed to be the second (it would have at least $17$ and at most $20$ points), so players should aim to achieve a better position among the runners-up, where matches played against the fifth (Poland or Israel) and the sixth (Azerbaijan) teams are discarded. If France defeats Israel, the latter would get the fifth place in Group 1, therefore France has $10$ points, $8$ goals for and $2$ goals against (a goal difference of $+6$) among the second-placed teams, according to Table~\ref{Table2}.

However, if France concedes two goals, Israel would be the fourth due to its $13$ points and better head-to-head result against Poland (see Table~\ref{Table2}): they have $3$ points and a goal difference of zero against each other, but Israel scored $3$ goals away, while Poland scored only $1$ goal away. Consequently, France would have $10$ points, $9$ goals for and $3$ goals against (a goal difference of $+6$) among the second-placed teams.
Since the two worst runners-up are involved in a play-off for qualification, but the first six directly qualify, it is strictly better to have the same number of points and the same goal difference with more goals scored in the ranking of second-placed teams.
Zidane should have kicked two own goals, or have agreed with his teammates to make some mistakes in defence as a draw of 2-2 is also preferred by Israel.

Our analysis is based only on information available at the time of the match France against Israel, so all coaches and players could have recognized the strange situation. However, probably the current paper gives the first description of this observation.

\section{Conclusions} \label{Sec4}

The actual result of the match France vs Israel remained 2-0.
Perhaps the players were honest, despite the risk involved in not conceding two goals.
Maybe they were well-informed: the difference between the two scenarios was marginal and could not influence whether France was among the top six runners-up qualifying automatically, although some matches were played later. %(Spain vs. Macedonia started at 21:30 in Group 2, Portugal vs. Republic of Ireland started at 22:00.  

Nevertheless, it makes no sense to deny that the shadow of a major outrage floated above 1996 UEFA European Championship qualification, identified to be the first incentive incompatible qualifying for UEFA European Championships \citep{Csato2017i}.
We hope the presented example is a clear warning for football governing bodies to organize strategy-proof qualifications in the future. They can follow the mechanism suggested by \citet{Csato2017i} to guarantee this crucial property.

%\clearpage

\section*{Acknowledgements}
\addcontentsline{toc}{section}{Acknowledgements}
\noindent
%I am grateful to S\'andor Boz\'oki for useful advices. \\
A reviewer provided valuable comments and suggestions on earlier drafts. \\
We are indebted to the \href{https://en.wikipedia.org/wiki/Wikipedia_community}{Wikipedia community} for contributing to our research by collecting information used in the paper. \\
%We are also grateful to \href{https://arxiv.org/}{arXiv} for the quick publication of this paper. \\
The research was supported by OTKA grant K 111797 and by the MTA Premium Post Doctorate Research Program. %\\
%This research was partially supported by Pallas Athene Domus Scientiae Foundation. The views expressed are those of the author's and do not necessarily reflect the official opinion of Pallas Athene Domus Scientiae Foundation.

%\bibliographystyle{apalike}
%\addcontentsline{toc}{section}{References}
%\bibliography{All_references}

\end{document}